\documentclass[runningheads,a4paper]{llncs}

\usepackage{amssymb}
\usepackage{graphicx}
\newcommand{\negA}{\vspace{-0.05in}}
\newcommand{\negB}{\vspace{-0.1in}}
\newcommand{\negC}{\vspace{-0.18in}}

\newcommand{\mysection}[1]{\negC\section{#1}\negA}
\newcommand{\mysubsection}[1]{\negB\subsection{#1}\negA}

\newcommand{\myparagraph}[1]{\par\smallskip\par\noindent{\bf{}#1:~}}
\usepackage{algorithm}
\usepackage{algorithmic}
\newcommand{\alg}[1]{\mbox{\sf #1}}  

\usepackage{amsopn}

\newcommand{\comment}[1]{}

\usepackage{url}
\urldef{\mails}\path|{hadas, meizeh}@cs.technion.ac.il|

\begin{document}

\mainmatter

\title{Parameterized Algorithms for
Graph Partitioning Problems}


\author{Hadas Shachnai \and Meirav Zehavi}

\institute{Department of Computer Science, Technion, Haifa 32000, Israel\\
\mails}

\maketitle

\vspace{-1em}

\begin{abstract}
We study a broad class of graph partitioning problems,
where each problem
 is specified by a
graph $G\!=\!(V,\!E)$, and parameters $k$ and $p$. We seek a subset $U\!\subseteq\!V$ of size $k$, such
that $\alpha_1m_1 \!+\! \alpha_2m_2$ is at most (or at least) $p$,
 where $\alpha_1,\!\alpha_2\!\in\!\mathbb{R}$ are constants defining the problem, and
$m_1, m_2$ are the cardinalities of the edge sets having both endpoints, and exactly one endpoint,  in $U$, respectively.
This class of {\em fixed cardinality graph partitioning problems (FGPP)}
encompasses 
 {\sc Max $(k,\!n\!-\!k)$-Cut}, {\sc Min $k$-Vertex Cover}, 
{\sc $k$-Densest Subgraph}, and {\sc $k$-Sparsest Subgraph}.

Our main result is an $O^*\!(4^{k+o(k)}\!\Delta^k)$ algorithm for any problem in this class, 
where $\Delta \!  \geq  \! 1$ is the maximum degree in the input graph.
This resolves an open
question posed by Bonnet et al. [IPEC 2013]. 
We obtain faster algorithms for certain subclasses of FGPPs,
parameterized by $p$, or by $(k+p)$.
In particular, we give an  $O^*\!(4^{p+o(p)})$ time algorithm for {\sc Max $(k,n\!-\!k)$-Cut},
thus improving significantly
the best known $O^*\!(p^p)$ time algorithm.
\comment{
Given a graph $G\!=\!(V,\!E)$, and parameters $k$ and $p$, a {\em Local Graph Partitioning Problem (LGPP)}
 asks if there is a subset $U\!\subseteq\!V$ of size $k$, such that $\alpha_1m_1 \!+\! \alpha_2m_2$ is at most (or at least) $p$,
 where $\alpha_1,\!\alpha_2\!\in\!\mathbb{R}$ are constants defining the problem, $m_1$ is the number of edges
in $E$ having both endpoints in $U$, and $m_2$ is the number of edges in $E$ having exactly one endpoint in $U$.
Bonnet et al. [IPEC 2013] posed as an open question the existence of constants $a$ and $b$ such that {\em any}
LGPP can be solved in time $O^*\!(a^k\!\Delta^{bk})$, where $\Delta$ is the maximum degree of $G$.
In this paper, we answer this question affirmatively by developing an $O^*\!(4^{k+o(k)}\!\Delta^k)$
 time algorithm for any LGPP. For certain subfamilies of LGPPs, we also develop algorithms parameterized by
$p$ and $(k+p)$. In particular, we solve a maximization LGPP called {\sc Max $(k,n\!-\!k)$-Cut} in time
$O^*\!(4^{p+o(p)})$, significantly improving the previous $O^*\!(p^p)$ time.
 Furthermore, we show that any minimization LGPP satisfying $\alpha_1\!\geq\!0$ and $\alpha_2\!>\!0$ can be solved
in time $O^*\!(2^{k+\frac{p}{\alpha_2}+o(k+p)})$, and develop a faster algorithm for such LGPPs
satisfying also $\alpha_2\!\geq\!\frac{\alpha_1}{2}$. This yields an $O^*\!(2^{p+o(p)})$ time algorithm
 for a minimization LGPP called {\sc Min $k$-Vertex Cover}, improving the previous {\em randomized} $O^*\!(3^p)$ time.
}
\end{abstract}

\vspace{-0.8em}

\mysection{Introduction}

\vspace{-0.1em}

Graph partitioning problems arise in many areas including VLSI design, data mining, parallel computing,
and sparse matrix factorizations (see, e.g., \cite{B06,KLM11,DRLJ10}).
We study a broad class of graph partitioning problems, where each problem
 is specified by a
graph $G\!=\!(V,\!E)$, and parameters $k$ and $p$. We seek a subset $U\!\subseteq\!V$ of size $k$, such
that $\alpha_1m_1 \!+\! \alpha_2m_2$ is at most (or at least) $p$,
 where $\alpha_1,\!\alpha_2\!\in\!\mathbb{R}$ are constants defining the problem, and
$m_1, m_2$ are the cardinalities of the edge sets having both endpoints, and exactly one endpoint,  in $U$, respectively.
This class encompasses such fundamental 
problems as {\sc Max} and {\sc Min $(k,\!n\!-\!k)$-Cut}, {\sc Max} and {\sc Min $k$-Vertex Cover}, 
{\sc $k$-Densest Subgraph}, and {\sc $k$-Sparsest Subgraph}.
For example, {\sc Max $(k,\!n\!-\!k)$-Cut} is a max-FGPP 
(i.e., maximization FGPP) 
satisfying $\alpha_1\!=\!0$ 
and $\alpha_2\!=\!1$, {\sc Min $(k,\!n\!-\!k)$-Cut} is a min-FGPP 
(i.e., minimization FGPP) 
satisfying $\alpha_1\!=\!0$
 and $\alpha_2\!=\!1$, and {\sc Min $k$-Vertex Cover} is a min-FGPP satisfying $\alpha_1\!=\!\alpha_2\!=\!1$.

A parameterized algorithm with parameter $k$ has running time $O^*\!(\!f(k))$ for some
 function $f$, where $O^*$ hides factors polynomial in the input size. In this paper, we develop a parameterized 
algorithm with parameter $(k+\Delta)$ for the class of all FGPPs, where $\Delta \! \geq \! 1$ is the maximum degree in the graph 
$G$. For certain subclasses of FGPPs, we develop algorithms parameterized by $p$, or by $(k+p)$.

\myparagraph{Related Work}Parameterized by $k$, {\sc Max} and {\sc Min $(\!k,\!n\!-\!k\!)$-Cut}, and {\sc Max} and
{\noindent {\sc Min $k$-Vertex Cover} are W[1]-hard \cite{downey03,cai08,vcvariants}. Moreover, 
{\sc $k$-Clique} and {\sc $k$-Independent Set}, two well-known W[1]-hard problems \cite{downey95}, 
are special cases of {\sc $k$-Densest Subgraph} where $p\!=\!k(k\!-\!1)$, and {\sc $k$-Sparsest Subgraph}
 where $p\!=\!0$, respectively. Therefore, parameterized by $(k\!+\!p)$, {\sc $k$-Densest Subgraph} and {\sc $k$-Sparsest Subgraph} are W[1]-hard.
 Cai et al. \cite{randsep} and Bonnet et al. \cite{localpartition} studied the parameterized complexity of FGPPs with respect to $(k\!+\!\Delta)$. Cai et al. \cite{randsep} gave $O^*\!(2^{(k\!+\!1)\Delta})$ time algorithms for {\sc $k$-Densest Subgraph} and 
{\sc $k$-Sparsest Subgraph}. Recently, Bonnet et al. \cite{localpartition} presented an $O^*\!(\Delta^k)$ time algorithm for 
{\em degrading} FGPPs. This subclass includes max-FGPPs in which $\alpha_1\!/2\!\leq\! \alpha_2$, and min-FGPPs 
in which $\alpha_1\!/2\!\geq\! \alpha_2$.\footnote{A max-FGPP (min-FGPP) is non-degrading if $\alpha_1/2\geq \alpha_2$ ($\alpha_1/2\leq \alpha_2$).} 
They also proposed  an $O^*\!(k^{2k}\Delta^{2k})$ time algorithm for all FGPPs, 
and posed as an open question the existence of constants $a$ and $b$ such that any FGPP can
 be solved in time $O^*\!(a^k\!\Delta^{bk})$. In this paper we answer this question affirmatively, by
developing an $O^*\!(4^{k+o(k)}\!\Delta^k\!)$ time algorithm for any FGPP.}

Parameterized by $p$, {\sc Max} and {\sc Min $k$-Vertex Cover} can be solved~in~times $O^*\!(1.396^p)$ and $O^*\!(4^p)$, respectively, and in randomized times $O^*\!(1.2993^p)$ and $O^*\!(3^p)$, respectively \cite{partialvc}. Moreover, {\sc Max $(k,\!n\!-\!k)$ Cut} can be solved in time $O^*\!(p^p)$ \cite{localpartition}, and {\sc Min $(k,\!n\!-\!k)$ Cut} can be solved in time $O(2^{O(p^3)})$ \cite{bisection}. Parameterized by $(k\!+\!p)$, {\sc Min $(k,\!n\!-\!k)$ Cut} can be solved in time $O^*\!(k^{2k}(k\!+\!p)^{2k}\!)$~\cite{localpartition}.

We note that the parameterized complexity of FGPPs has also been studied with respect to other parameters, such as the treewidth and the vertex cover number of $G$ (see, e.g., \cite{kloks94,kdensest,localpartition}).

\myparagraph{Contribution}Our main result  is an $O^*\!(4^{k+o(k)}\!\Delta^k\!)$ time 
algorithm for~the~class of all FGPPs, answering affirmatively the question posed by Bonnet~et~al.~\cite{localpartition}
(see Section \ref{section:res1}).
In Section \ref{section:res2a}, we develop an $O^*\!(4^{p+o(p)})$ time algorithm for {\sc Max $(k,\!n\!-\!k)$-Cut}, 
that significantly improves the $O^*\!(p^p)$ running time obtained in \cite{localpartition}. 
We also obtain (in Section \ref{section:res2b}) an $O^*\!(2^{k+\frac{p}{\alpha_2}+o(k+p)})$ time algorithm
 for the subclass of {\em positive} min-FGPPs, in which $\alpha_1\!\geq\!0$ and $\alpha_2\!>\!0$.  Finally, we develop (in Section \ref{section:res3}) a faster algorithm for non-degarding positive min-FGPPs (i.e., min-FGPPs satisfying $\alpha_2\!\geq\!\frac{\alpha_1}{2}\!>\!0$). In particular, we thus solve {\sc Min $k$-Vertex Cover} in time $O^*\!(2^{p+o(p)})$, 
improving the previous {\em randomized} $O^*\!(3^p)$ time algorithm.

\myparagraph{Techniques}
We obtain our main result by establishing 
an interesting reduction from non-degrading FGPPs to
 the {\sc Weighted $k'$-Exact Cover ($k'$-WEC)} problem (see Section \ref{section:res1}).
 Building on this reduction, combined with an algorithm for degrading FGPPs given in \cite{localpartition}, and 
an algorithm for {\sc $k'$-WEC} given in \cite{corrmatchpack}, we develop an algorithm for any FGPP. 
To improve the running time of our algorithm, we use a fast
construction of representative families~\cite{representative,corrrepresentative}.

In designing algorithms for FGPPs, parameterized by $p$ or $(k+p)$, we use as a key tool 
{\em randomized separation} \cite{randsep} (see Sections \ref{section:res2a}--\ref{section:res3}). 
Roughly speaking, randomized separation finds a `good' partition of the nodes in the input graph $G$ via randomized
coloring of the nodes in {\em red} or {\em blue}.
 If a solution exists, then, with some positive probability, there is a set $X$ of {\em only} red nodes that is a solution,
 such that {\em all} the neighbors of nodes in $X$ that are outside $X$ are blue. 
Our algorithm for {\sc Max $(k,\!n\!-\!k)$-Cut} makes non-standard use of randomized separation,
 in requiring that only {\em some} of the neighbors outside $X$ of nodes in $X$ are blue. 
This yields the desired improvement in the running time of our algorithm.

Our algorithm for non-degrading positive FGPPs is based on a somewhat different application of randomized separation,
in which we randomly color {\em edges} rather than the nodes. If a solution exists, then, with some positive probability, there is a node-set $X$ that is a solution, such that {\em some} edges between nodes in $X$ are red, and {\em all} edges between nodes in $X$ and nodes outside $X$ are blue. In particular, we require
that the subgraph induced by $X$, and the subgraph induced by $X$ from which we delete all blue edges, 
contain the same connected components.
We derandomize our algorithms using universal sets \cite{splitter}.

\myparagraph{Notation}Given a graph $G\!=\!(V,\!E)$ and a subset $X\!\subseteq\! V$, let $E(X)$ denote the set of edges in $E$ having both endpoints in $X$, and let $E(X,\!V\!\setminus\! X)$ denote the set of edges in $E$ having exactly one endpoint in $X$. Moreover, given a subset $X\!\subseteq\!V$, let $\mathrm{val}(X)\!=\! \alpha_1|E(X)|\!+\!\alpha_2|E(X,\!V\!\setminus\!X)|$.

\mysection{Solving FGPPs in Time $O^*(4^{k+o(k)}\Delta^k)$}\label{section:res1}

\vspace{-0.6em}

In this section we develop an $O^*\!(4^{k+o(k)}\Delta^k)$ time algorithm for the class of all FGPPs.
We use the following steps. In Section \ref{section:reduce} we show  that any non-degrading FGPP 
can be reduced to the {\sc Weighted $k'$-Exact Cover ($k'$-WEC)} problem, where $k'=k$. Applying this reduction, we then
 show (in Section \ref{section:decrease}) how to decrease the size of instances of {\sc $k'$-WEC}, by using representative families. Finally, we show (in Section \ref{section:algorithm}) how to solve any FGPP by using the results in Sections \ref{section:reduce} and \ref{section:decrease}, an algorithm for {\sc $k'$-WEC}, and an algorithm for degrading FGPPs given in \cite{localpartition}.

\vspace{-0.6em}

\mysubsection{From Non-Degrading FGPPs to {\sc $k'$-WEC}}\label{section:reduce}

\vspace{-0.3em}

We show below that any non-degrading max-FGPP can be reduced to the maximization version of {\sc $k'$-WEC}. Given a universe $U$, a family ${\cal S}$ of nonempty subsets of $U$, a function $w\!: {\cal S}\!\rightarrow\!\mathbb{R}$, and parameters $k'\!\in\!\mathbb{N}$ and $p'\!\in\!\mathbb{R}$,
 we seek a subfamily ${\cal S}'$ of disjoint sets from ${\cal S}$ satisfying
 $|\bigcup {\cal S}'|=k'$ whose value, given by $\sum_{S\in{\cal S}'}w(S)$, is at least $p'$.
Any non-degrading min-FGPP can be similarly reduced to the minimization version of {\sc $k'$-WEC}.

Let $\Pi$ be a max-FGPP satisfying $\frac{\alpha_1}{2}\!\geq\!\alpha_2$. Given an instance ${\cal I}\!=\!(G\!=\!(V,\!E),\!k,\!p)$ of $\Pi$, we define an instance $f({\cal I})\!=\!(U,\!{\cal S},\!w,\!k',\!p')$ of the maximization version of {\sc $k'$-WEC} as follows.

\vspace{-0.5em}
\begin{itemize}
\item $U\!=\!V$.
\item ${\cal S}\!=\!\bigcup_{i=1}^k{\cal S}_i$, where ${\cal S}_i$ contains the node-set of any connected subgraph of $G$ on exactly $i$ nodes.
\item $\forall S\!\in\!{\cal S}: w(S) = \mathrm{val}(S)$.
\item $k'\!=\!k$, and $p'\!=\!p$.
\end{itemize}
\vspace{-0.1em}

{\noindent We illustrate the reduction in Figure~\ref{fig:1} (see Appendix~\ref{app:fig1}). We first prove that our reduction is valid.}
\begin{lemma}\label{lemma:reducecor}
${\cal I}$ is a yes-instance iff $f({\cal I})$ is a yes-instance.
\end{lemma}
\begin{proof}
First, assume there is a subset $X\!\subseteq\! V$ of size $k$ satisfying $\mathrm{val}(X)\!\geq\!p$. Let $G_1\!=\!(V_1,\!E_1),\ldots,G_t\!=\!(V_t,\!E_t)$, for some $1\!\leq\! t\!\leq\! k$, be the {\em maximal} connected components in the subgraph of $G$ induced by $X$. Then, for all $1\!\leq\! \ell\!\leq\! t$, $V_{\ell}\!\in\!{\cal S}$. Moreover, $\displaystyle{\sum_{\ell=1}^t|V_{\ell}|\!=\!|X|\!=\!k'}$, and $\displaystyle{\sum_{\ell=1}^tw(V_{\ell})\!=\!\mathrm{val}(X)\!\geq\!p'}$.

Now, assume there is a subfamily of disjoint sets $\{S_1,\ldots,S_t\}\!\subseteq\!{\cal S}$, for some $1\!\leq\! t\!\leq\! k$, such that $\displaystyle{\sum_{\ell=1}^t|S_{\ell}|\!=\!k'}$ and $\displaystyle{\sum_{\ell=1}^tw(S_{\ell})\!\geq\!p'}$. Thus, there are connected subgraphs $G_1\!=\!(V_1,\!E_1),\ldots,G_t\!=\!(V_t,\!E_t)$ of $G$, such that $V_{\ell}\!=\!S_{\ell}$, for all $1\!\leq\! \ell\!\leq\! t$.
 Let $X_{\ell}\!=\!\bigcup_{j=\ell}^tV_j$, for all $1\!\leq \ell\!\leq\! t$. Clearly, $|X_1|\!=\!k$. Since $\frac{\alpha_1}{2}\!\geq\!\alpha_2$, we get~that

\[\begin{array}{ll}

\mathrm{val}(X_1) &= \mathrm{val}(V_1) \!+\! \mathrm{val}(X_2) \!+\! \alpha_1|E(V_1,X_2)| \!-\! 2\alpha_2|E(V_1,X_2)|\\

&\geq \mathrm{val}(V_1) \!+\! \mathrm{val}(X_2)\\

&= \mathrm{val}(V_1) \!+\! \mathrm{val}(V_2) \!+\! \mathrm{val}(X_3) \!+\! \alpha_1|E(V_2,X_3)| \!-\! 2\alpha_2|E(V_2,X_3)|\\

&\geq \mathrm{val}(V_1) \!+\! \mathrm{val}(V_2) \!+\! \mathrm{val}(X_3)\\

& ...\\

&\geq \displaystyle{\sum_{\ell=1}^t\mathrm{val}(V_{\ell})}.

\end{array}\]
Thus, $\displaystyle{\mathrm{val}(X_1)\!\geq\!\sum_{\ell=1}^tw(V_{\ell})\!\geq\!p}$.\qed
\end{proof}
We now bound the number of connected subgraphs in $G$.
\begin{lemma}[\cite{enuconcom}]
There are at most $4^i(\Delta\!-\!1)^i|V|$ connected subgraphs of $G$ on at most $i$ nodes, which can be enumerated in time $O(4^i(\Delta\!-\!1)^i(|V|\!+\!|E|)|V|)$.
\end{lemma}
Thus, we have the next result.
\begin{lemma}\label{lemma:reducetim}
The instance $f({\cal I})$ can be constructed in time $O(4^k(\Delta\!-\!1)^k(|V|\!+\!|E|)|V|)$. Moreover, for any $1\!\leq\!i\!\leq\!k$, $|{\cal S}_i|\!\leq\! 4^i(\Delta\!-\!1)^i|V|$.
\end{lemma}

\vspace{-0.2em}

\mysubsection{Decreasing the Size of Inputs for {\sc $k'$-WEC}}\label{section:decrease}

In this section we develop a procedure, called \alg{Decrease}, which 
decreases the size of  an instance 
$\!(U,\!{\cal S},\!w,\!k',\!p')$ of $k'$-WEC. 
To this end, we find a subfamily $\widehat{\cal S}\!\subseteq\!{\cal S}$ that contains "enough" sets
 from ${\cal S}$, and thus enables to replace ${\cal S}$ by $\widehat{\cal S}$ without turning a yes-instance to a no-instance. The following definition captures such a subfamily $\widehat{\cal S}$.

\begin{definition}
Given a universe $E$, nonnegative integers $k$ and $p$, a family ${\cal S}$ of subsets of size $p$ of $E$, and a function $w\!:\!{\cal S}\!\rightarrow\!\mathbb{R}$, we say that a subfamily $\widehat{\cal S}\!\subseteq\!{\cal S}$ {\em max (min) represents} $\cal S$ if for any pair of sets $X\!\in\!{\cal S}$, and $Y\!\subseteq\!E\!\setminus\!X$ such that $|Y|\!\leq\!k\!-\!p$, there is a set $\widehat{X}\!\in\!\widehat{\cal S}$ disjoint from $Y$ such that $w(\!\widehat{X}\!)\!\geq\!w(\!X\!)$ ($w(\!\widehat{X}\!)\!\leq\!w(\!X\!)$).
\end{definition}
The following result states that small representative families can be computed
 efficiently.\footnote{This result builds on a powerful construction technique for representative families 
presented in \cite{representative}.}

\begin{theorem}[\cite{corrrepresentative}]\label{theorem:rep}
Given a constant $c\!\geq\!1$, a universe $E$, nonnegative integers $k$ and $p$, a family ${\cal S}$ of subsets of size $p$ of $E$, and a function $w\!:\!{\cal S}\!\rightarrow\!\mathbb{R}$, a subfamily $\widehat{\cal S}\!\subseteq\!{\cal S}$ of size at most $\displaystyle{\frac{(ck)^{k}}{p^p(ck\!-\!p)^{k\!-\!p}}2^{o(k)}\!\log\!|E|}$ that max (min) represents $\cal S$ can be computed in time $\displaystyle{O(\!|{\cal S}|(ck/(ck\!-\!p))^{k\!-\!p}2^{o(k)}\!\log\!|E|\! +\! |{\cal S}|\!\log\!|{\cal S}|)}$.
\end{theorem}
We next consider the maximization version of $k'$-WEC and max representative families. The minimization version of $k'$-WEC can be similarly handled by using min representative families. Let \alg{RepAlg}$(E,\!k,\!p,\!{\cal S},\!w)$ denote the algorithm in Theorem \ref{theorem:rep} where $c\!=\!2$, and let ${\cal S}_i\!=\!\{S\!\in\!{\cal S}\!: |S|\!=\!i\}$, for all $1\!\leq\!i\!\leq\!k'$.

We now present procedure \alg{Decrease} (see the pseudocode below), which  replaces each family ${\cal S}_i$ by a family $\widehat{\cal S}_i\!\subseteq\!{\cal S}_i$ that represents ${\cal S}_i$. First, we state that procedure \alg{Decrease} is correct (the proof is given in Appendix~\ref{app:someproofs}).

\renewcommand{\thealgorithm}{}
\floatname{algorithm}{Procedure}
\begin{algorithm}[!ht]
\caption{\alg{Decrease}($U,\!{\cal S},\!w,\!k',\!p'$)}
\begin{algorithmic}[1]
\STATE {\bf for} $i=1,2,\ldots,k'$ {\bf do} $\widehat{\cal S}_i\Leftarrow$ \alg{RepAlg}($U,\!k',\!i,\!{\cal S}_i,\!w$). {\bf end for}
\STATE $\widehat{\cal S}\Leftarrow \bigcup_{i=1}^k\widehat{\cal S}_i$.
\STATE {\bf return} ($U,\!\widehat{\cal S},\!w,\!k',\!p'$).
\end{algorithmic}
\end{algorithm}

\begin{lemma}\label{lemma:decreasecor}
$(U,\!{\cal S},\!w,\!k',\!p')$ is a yes-instance iff $(U,\!\widehat{\cal S},\!w,\!k',\!p')$ is a yes-instance.
\end{lemma}

\comment{
\begin{proof}
First, assume that $(U,\!{\cal S},\!w,\!k',\!p')$ is a yes-instance. Let ${\cal S}'$ be 
a subfamily of disjoint sets from ${\cal S}$, such that $|\bigcup{\cal S}'|\!=\!k'$, $\sum_{S\in{\cal S}'}w(S)\!\geq\!p'$, 
and there is no subfamily ${\cal S}''$ satisfying these conditions, and
 $|{\cal S}'\!\cap\!\widehat{\cal S}|\!<\!|{\cal S}''\!\cap\!\widehat{\cal S}|$. Suppose, by way of
 contradiction, that there is a set $S\!\in\!({\cal S}_i\cap{\cal S}')\!\setminus\!\widehat{\cal S}$,
 for some $1\!\leq\!i\!\leq\!k'$. By Theorem \ref{theorem:rep}, there is a set $\widehat{S}\!\in\!\widehat{\cal S}_i$
 such that $w(\widehat{S})\!\geq\!w(S)$, and $\widehat{S}\!\cap\!S'\!=\!\emptyset$,
 for all $S'\!\in\!{\cal S}'\!\setminus\!\{S\}$. Thus, ${\cal S}''\!=\!({\cal S}'\!\setminus\!\{S\})\!\cup\!\{\widehat{S}\}$
 is a solution to $(U,\!{\cal S},\!w,\!k',\!p')$. Since $|{\cal S}'\!\cap\!\widehat{\cal S}|\!<\!|{\cal S}''\!\cap\!\widehat{\cal S}|$,
 this is a contradiction.

Now, assume that $(U,\!\widehat{\cal S},\!w,\!k',\!p')$ is a yes-instance. Since $\widehat{\cal S}\!\subseteq\!{\cal S}$, 
we immediately get that $(U,\!{\cal S},\!w,\!k',\!p')$ is also a yes-instance.\qed
\end{proof}
}
{\noindent Theorem \ref{theorem:rep} immediately implies the following result.}

\begin{lemma}\label{lemma:decreasetim}
Procedure \alg{Decrease} runs in time $\displaystyle{O(\sum_{i=1}^{k'}(|{\cal S}_i|(\frac{2k'}{2k'\!-\!i})^{k'-i}2^{o(k')}\!\log\!|U|}$\\$+ |{\cal S}_i|\!\log\!|{\cal S}_i|))$. Moreover, $\displaystyle{|\widehat{\cal S}| \leq \sum_{i=1}^{k'}\frac{(2k')^{k'}}{i^i(2k'\!-\!i)^{k'-i}}2^{o(k')}\!\log\!|U| \leq  2.5^{k'+o(k')}\!\log\!|U|}$.
\end{lemma}

\vspace{-0.2em}

\mysubsection{An Algorithm for Any FGPP}\label{section:algorithm}

We now present \alg{FGPPAlg}, which solves any FGPP in time $O^*(4^{k+o(k)}\Delta^k)$. Assume w.l.o.g that $\Delta\!\geq\!2$, and let \alg{DegAlg}($G,\!k,\!p$) denote the algorithm solving any degrading FGPP in time $O((\Delta\!+\!1)^{k+1}|V|)$, given in~\cite{localpartition}.

The algorithm given in Section 5 of \cite{corrmatchpack} solves a problem closely related to $k'$-WEC, and can be 
easily modified to solve $k'$-WEC in time $O(2.851^{k'}|{\cal S}||U|\cdot$\\$\log^2|U|)$.
 We call this algorithm \alg{WECAlg}($U,\!{\cal S},\!w,\!k',\!p'$). 

Let $\Pi$ be an FGPP having parameters $\alpha_1$ and $\alpha_2$. We now describe algorithm \alg{FGPPAlg} (see the pseudocode below). 
First, if $\Pi$ is a degrading FGPP,  then \alg{FGPPAlg} solves $\Pi$ by calling \alg{DegAlg}.
Otherwise, by using the reduction $f$, \alg{FGPPAlg} transforms the input into an instance of {\sc $k'$-WEC}. Then, \alg{FGPPAlg} decreases the size of the resulting
 instance by calling the procedure \alg{Decrease}. Finally, \alg{FGPPAlg} solves $\Pi$ by calling \alg{WECAlg}.

\renewcommand{\thealgorithm}{1}
\floatname{algorithm}{Algorithm}
\begin{algorithm}[!ht]
\caption{\alg{FGPPAlg}($G\!=\!(V,\!E),\!k,\!p$)}
\begin{algorithmic}[1]
\IF{($\Pi$ is a max-FGPP and $\frac{\alpha_1}{2}\!\leq\! \alpha_2$) or ($\Pi$ is a min-FGPP and $\frac{\alpha_1}{2}\!\geq\! \alpha_2$)}
	\STATE {\bf accept} iff \alg{DegAlg}$(G,\!k,\!p)$ accepts.
\ENDIF
\STATE $(U,\!{\cal S},\!w,\!k',\!p')\Leftarrow f(G,\!k,\!p)$.
\STATE $(U,\!\widehat{\cal S},\!w,\!k',\!p')\Leftarrow$ \alg{Decrease}$(U,\!{\cal S},\!w,\!k',\!p')$.
\STATE {\bf accept} iff \alg{WECAlg}($U,\!\widehat{\cal S},\!w,\!k',\!p'$) accepts.
\end{algorithmic}
\end{algorithm}

\begin{theorem}
Algorithm \alg{FGPPAlg} solves $\Pi$ in time $O(4^{k+o(k)}\Delta^k(|V|\!+\!|E|)|V|)$.
\end{theorem}

\begin{proof}
The correctness of the algorithm follows immediately from Lemmas \ref{lemma:reducecor} and \ref{lemma:decreasecor}, and the correctness of \alg{DegAlg} and \alg{WECAlg}.

By Lemmas \ref{lemma:reducetim} and \ref{lemma:decreasetim}, and the running times of \alg{DegAlg} and \alg{WECAlg}, algorithm \alg{FGPPAlg} runs in time

\vspace{-0.5em}

\[\begin{array}{lll}

  &  &\displaystyle{O(4^k(\Delta\!-\!1)^k(|V|\!+\!|E|)|V| + \sum_{i=1}^k(4^i(\Delta\!-\!1)^i|V|(\frac{2k}{2k-i})^{k-i}2^{o(k)}\!\log\!|V|)}\\
  &  & \displaystyle{+\ 2.851^k2.5^{k+o(k)}|V|\log^3|V|)}\\

  &= &\displaystyle{O(4^k\Delta^k(|V|\!+\!|E|)|V| + 2^{o(k)}|V|\!\log\!|V|[\max_{0\leq\alpha\leq1}\{4^{\alpha}\Delta^{\alpha}(\frac{2}{2-\alpha})^{1-\alpha}\}]^k)}\\

   &= &\displaystyle{O(4^k\Delta^k(|V|\!+\!|E|)|V| + 4^{k+o(k)}\Delta^k|V|\!\log\!|V|)}\\

  &= &\displaystyle{O(4^{k+o(k)}\Delta^k(|V|\!+\!|E|)|V|)}.
\end{array}\]
\vspace{-0.5em}\qed
\end{proof}

\mysection{Solving {\sc Max $(k,n\!-\!k)$ Cut} in Time $O^*\!(4^{p+o(p)})$}\label{section:res2a}

We give below an $O^*\!(4^{p+o(p)})$ time algorithm for {\sc Max $(k,n\!-\!k)$ Cut}.
In Section \ref{section:mcsimple} we show that it suffices to consider an easier variant of
{\sc Max $(k,n\!-\!k)$ Cut}, that we call {\sc NC-Max $(k,\!n\!-\!k)$-Cut}.
We solve this variant in Section \ref{section:ncmcut}. Finally, our algorithm for
{\sc Max $(k,n\!-\!k)$ Cut} is given in Section \ref{section:cutalg}.

\vspace{-0.2em}

\mysubsection{Simplifying {\sc Max $(k,n\!-\!k)$ Cut}}\label{section:mcsimple}

We first define an easier variant of {\sc Max $(k,\!n\!-\!k)$ Cut}. Given a graph $G\!=\!(V,\!E)$ in which each node is
 either red or blue, and positive integers $k$ and $p$,
 {\sc NC-Max $(k,\!n\!-\!k)$-Cut} asks if there is a subset $X\!\subseteq\! V$ of exactly $k$ red nodes and no blue nodes,
 such that at least $p$ edges in $E(X,\!V\!\setminus\! X)$ have a blue endpoint.

Given an instance $(G,\!k,\!p)$ of {\sc Max $(k,\!n\!-\!k)$ Cut}, we perform several iterations of
 coloring the nodes in $G$; thus, if $(G,\!k,\!p)$ is a yes-instance, we generate at least one yes-instance of
 {\sc NC-Max $(k,\!n\!-\!k)$-Cut}. To determine how to color the nodes in $G$, we need the following definition of universal sets.

\vspace{-0.2em}

\begin{definition}
Let ${\cal F}$ be a set of functions $f\!:\! \{1,\!2,\!\ldots,\!n\}\rightarrow \{0,\!1\}$. We say that ${\cal F}$ is an $(n,t)$-universal set if for every subset $I\!\subseteq\!\{1,\!2,\!\ldots,\!n\}$ of size $t$ and a function $f'\!:\!I\!\rightarrow\!\{0,\!1\}$, there is a function $f\!\in\!{\cal F}$ such that for all $i\!\in\!I$, $f(i)\!=\!f'(i)$.
\end{definition}

\vspace{-0.2em}

{\noindent The following result asserts that small universal sets can be computed efficiently.}

\vspace{-0.2em}

\begin{lemma}[\cite{splitter}]\label{lemma:splitter}
There is an algorithm, \alg{UniSetAlg}, that given a pair of integers $(n,\!t)$, computes an $(n,\!t)$-universal set ${\cal F}$ of size $2^{t\!+\!o(t)}\!\log\! n$ in time $O\!(2^{t\!+\!o(t)}n\!\log\!n\!)$.
\end{lemma}

\vspace{-0.2em}

{\noindent We now present \alg{ColorNodes} (see the pseudocode below), a procedure that given an input ($G,\!k,\!p,\!q$), where ($G,\!k,\!p$) is an instance of {\sc Max $(k,\!n\!-\!k)$ Cut} and $q\!=\!k+p$, returns a set of instances of {\sc NC-Max $(k,\!n\!-\!k)$-Cut}. Procedure \alg{ColorNodes} first constructs a $(|V|,\!k\!+\!p)$-universal set $\cal F$. For each $f\!\in\!{\cal F}$, \alg{ColorNodes} generates a colored copy $V^f$ of $V$. Then, \alg{ColorNodes} returns a set $\cal I$, including the resulting instances of {\sc NC-Max $(k,\!n\!-\!k)$-Cut}.}

\renewcommand{\thealgorithm}{}
\floatname{algorithm}{Procedure}
\begin{algorithm}[!ht]
\caption{\alg{ColorNodes}($G\!=\!(V,E),k,p,q$)}
\begin{algorithmic}[1]
\STATE let $V\!=\!\{v_1,v_2,\ldots,v_{|V|}\}$.
\STATE\label{step:colornodes2} ${\cal F}\Leftarrow$ \alg{UniSetAlg}$(|V|,q)$.

\FORALL{$f\!\in\!{\cal F}$}
	\STATE let $V^f\!=\!\{v^f_1,v^f_2,\ldots,v^f_{|V|}\}$, where $v^f_i$ is a copy of $v_i$.
	\FOR{$i=1,2,\ldots,|V|$}
		\STATE {\bf if} $f(i)\!=\!0$ {\bf then} color $v^f_i$ red. {\bf else} color $v^f_i$ blue. {\bf end if}
	\ENDFOR
\ENDFOR

\STATE {\bf return} ${\cal I}=\{(G_f\!=\!(V_f,E),k,p): f\!\in\!{\cal F}\}$.
\end{algorithmic}
\end{algorithm}

{\noindent The next lemma states the correctness of procedure \alg{ColorNodes}.}

\vspace{-0.2em}

\begin{lemma}\label{lemma:colornodescor}
An instance $(G,\!k,\!p)$ of {\sc Max $(k,\!n\!-\!k)$-Cut} is a yes-instance iff \alg{ColorNodes}($G,\!k,\!p,\!k\!+\!p$) returns a set ${\cal I}$ containing at least one yes-instance of {\sc NC-Max $(k,\!n\!-\!k)$-Cut}.
\end{lemma}

\vspace{-1em}

\begin{proof}
If $(G,\!k,\!p)$ is a no-instance of {\sc  Max $(k,\!n\!-\!k)$-Cut}, then clearly, for any coloring of the nodes in $V$, we get a no-instance of {\sc NC-Max $(k,\!n\!-\!k)$-Cut}. Next suppose that $(G,\!k,\!p)$ is a yes-instance, and let $X$ be a set of $k$ nodes in $V$ such that $|E(X,\!V\!\setminus\!X)|\!\geq\!p$. Note that there is a set $Y$ of at most $p$ nodes in $V\!\setminus\!X$ such that $|E(X,\!Y)|\!\geq\!p$. Let $X'$ and $Y'$ denote the indices of the nodes in $X$ and $Y$, respectively. Since ${\cal F}$ is a $(|V|,\!k\!+\!p)$-universal set, there is $f\!\in\!{\cal F}$ such that: (1) for all $i\!\in\!X'$, $f(i)\!=\!0$, and (2) for all $i\!\in\!Y'$, $f(i)\!=\!1$. Thus, in $G_f$, the copies of the nodes in $X$ are red, and the copies of the nodes in $Y$ are blue. We get that $(G_f,\!k,\!p)$ is a yes-instance of {\sc NC-Max $(k,\!n\!-\!k)$-Cut}.\qed
\end{proof}

\vspace{-0.1em}

{\noindent Furthermore, Lemma \ref{lemma:splitter} immediately implies the following result.}

\vspace{-0.2em}

\begin{lemma}\label{lemma:colornodestime}
Procedure \alg{ColorNodes} runs in time $O(2^{q+o(q)}|V|\log\!|V|)$, and returns a set ${\cal I}$ of size $O(2^{q+o(q)}\log\!|V|)$.
\end{lemma}

\vspace{-0.2em}

\mysubsection{A Procedure for {\sc NC-Max $(k,\!n\!-\!k)$-Cut}}\label{section:ncmcut}

\vspace{-0.2em}

We now present \alg{SolveNCMaxCut}, a procedure for solving {\sc NC-Max $(k,\!n\!-\!k)$-Cut}

{\noindent (see the pseudocode below). Procedure \alg{SolveNCMaxCut} orders the red nodes in $V$
 according to the number of their blue neighbors in a non-increasing manner. If there are at least $k$ red nodes,
 and the number of edges between the first $k$ red nodes and blue nodes is at least $p$, procedure \alg{SolveNCMaxCut}
 accepts; otherwise, procedure \alg{SolveNCMaxCut} rejects.}

\renewcommand{\thealgorithm}{}
\floatname{algorithm}{Procedure}
\begin{algorithm}[!ht]
\caption{\alg{SolveNCMaxCut}($G\!=\!(V,E),k,p$)}
\begin{algorithmic}[1]
\STATE {\bf for all} red $v\!\in\! V$ {\bf do} compute the number $n_b(v)$ of blue neighbors of $v$ in $G$. {\bf end~for}

\STATE let $v_1,\!v_2,\!\ldots,\!v_r$, for some $0\leq r\leq |V|$, denote the red nodes in $V$, such that $n_b(v_i)\geq n_b(v_{i+1})$ for all $1\!\leq\! i\!\leq\! r\!-\!1$.

\STATE {\bf accept} iff ($r\!\geq\!k$ and $\sum_{i=1}^kn_b(v_i)\!\geq\! p$).
\end{algorithmic}
\end{algorithm}

Clearly, the following result concerning \alg{SolveNCMaxCut} is correct.

\begin{lemma}\label{lemma:ncmcut}
Procedure \alg{SolveNCMaxCut} solves {\sc NC-Max $(k,\!n\!-\!k)$-Cut} in time $O(|V|\!\log|V|\!+\!|E|)$.
\end{lemma}

\mysubsection{An Algorithm for {\sc Max $(k,n\!-\!k)$ Cut}}\label{section:cutalg}

Assume w.l.o.g that $G$ has no isolated nodes.
Our algorithm, \alg{MaxCutAlg}, for {\sc Max $(k,n\!-\!k)$ Cut},  proceeds as follows.
 First, if $p\!<\!\min\{k,|V|\!-\!k\}$, 
then \alg{MaxCutAlg} accepts, and if $|V|\!-\!k < k$, then \alg{MaxCutAlg} calls itself
 with $|V|\!-\!k$ instead of $k$. Then, \alg{MaxCutAlg} calls  \alg{ColorNodes} to compute a
 set of instances of {\sc NC-Max $(k,\!n\!-\!k)$-Cut}, and accepts iff \alg{SolveNCMaxCut} 
accepts at least one of them.

\renewcommand{\thealgorithm}{2}
\floatname{algorithm}{Algorithm}
\begin{algorithm}[!ht]
\caption{\alg{MaxCutAlg}($G\!=\!(V,E),k,p$)}
\begin{algorithmic}[1]
\STATE\label{step:cut1} {\bf if} $p < \min\{k,|V|\!-\!k\}$ {\bf then} {\bf accept}. {\bf end if}

\STATE\label{step:cut2} {\bf if} $|V|\!-\!k < k$ {\bf then} {\bf accept} iff \alg{MaxCutAlg}($G,|V|\!-\!k,p$) accepts. {\bf end if}

\STATE\label{step:cut3} ${\cal I}\Leftarrow$ \alg{ColorNodes}($G,k,p,k\!+\!p$).

\FORALL{$(G',k',p')\in{\cal I}$}\label{step:cut4}
	\STATE\label{step:cut5} {\bf if} \alg{SolveNCMaxCut}$(G',k',p')$ accepts {\bf then} {\bf accept}. {\bf end if}
\ENDFOR

\STATE\label{step:cut6} {\bf reject}.
\end{algorithmic}
\end{algorithm}

{\noindent The next lemma implies the correctness of Step \ref{step:cut1} in \alg{MaxCutAlg}.

\begin{lemma}[\cite{localpartition}]\label{lemma:pvsk}
In a graph $G\!=\!(V,\!E)$ having no isolated nodes, there is a subset $X\!\subseteq\!V$ of size $k$ such that $|E(X,\!V\!\setminus\! X)|\geq\min\{k,\!|V|\!-\!k\}$.
\end{lemma}
Our main result is the following.

\begin{theorem}
Algorithm \alg{MaxCutAlg} solves {\sc Max $(k,\!n\!-\!k)$ Cut} in time $O(4^{p\!+\!o(p)}\cdot$\\$(|V|+\!|E|)\log^2|V|)$.
\end{theorem}

\begin{proof}
Clearly, ($G,k,p$) is a yes-instance iff ($G,|V|\!-\!k,p$) is a yes-instance. Thus, Lemmas \ref{lemma:colornodescor}, \ref{lemma:ncmcut} and \ref{lemma:pvsk} immediately imply the correctness of \alg{MaxCutAlg}.

Denote $m\!=\!\min\{k,\!|V|\!-\!k\}$. If $p \!<\! m$, then \alg{MaxCutAlg} runs in time $O(1)$. Next suppose that $p \!\geq\! m$. Then, by Lemmas \ref{lemma:colornodestime} and \ref{lemma:ncmcut}, \alg{MaxCutAlg} runs in time $O(2^{m+p+o(m+p)}(|V|\!+\!|E|)\log^2|V|)=O(4^{p+o(p)}(|V|\!+\!|E|)\log^2|V|)$.\qed
\end{proof}

\mysection{Solving Positive Min-FGPPs in Time $O^*\!(2^{k+\frac{p}{\alpha_2}+o(k+p)})$}\label{section:res2b}

Let $\Pi$ be a  min-FGPP satisfying $\alpha_1\!\geq\!0$ and $\alpha_2\!>\!0$. In this section we develop an $O^*\!(2^{k+\frac{p}{\alpha_2}+o(k+p)})$ time algorithm for $\Pi$.
 Using randomized separation, we show in Section \ref{section:possimple} that we can focus on an 
easier version of $\Pi$. We solve this version in Section \ref{section:ncpos}, using dynamic programming. Then, 
Section \ref{section:posalg} gives our algorithm.

\mysubsection{Simplifying the Positive Min-FGPP $\Pi$}\label{section:possimple}

We first define an easier variant of $\Pi$. Given a graph $G\!=\!(V,\!E)$ in which each node is either red or blue, and parameters $k\!\in\!\mathbb{N}$ and $p\!\in\!\mathbb{R}$, {\sc NC-$\Pi$} asks if there is a subset $X\!\subseteq\! V$ of exactly $k$ red nodes and no blue nodes, whose neighborhood outside $X$ includes only blue nodes, such that val$(X)\!\leq\!p$.

The simplification process is similar to that performed in Section \ref{section:mcsimple}. However, we now use the randomized
 separation procedure \alg{ColorNodes}, defined in Section \ref{section:mcsimple}, with instances of $\Pi$, and consider the set ${\cal I}$ returned by \alg{ColorNodes} as a set of instances of {\sc NC-$\Pi$}. We next prove that \alg{ColorNodes} is correct.

\begin{lemma}\label{lemma:colornodesPcor}
An instance $(G,\!k,\!p)$ of {\sc $\Pi$} is a yes-instance iff \alg{ColorNodes}($G,\!k,\!p,\!k\!+\!\frac{p}{\alpha_2}$) returns a set ${\cal I}$ containing at least one yes-instance of {\sc NC-$\Pi$}.
\end{lemma}

\begin{proof}
If $(G,\!k,\!p)$ is a no-instance of {\sc $\Pi$}, then clearly, for any coloring of the nodes in $V$, we get a no-instance of {\sc NC-$\Pi$}. Next suppose that $(G,\!k,\!p)$ is a yes-instance, and let $X$ be a set of $k$ nodes in $V$ such that val$(X)\!\leq\!p$. Let $Y$ denote the neighborhood of $X$ outside $X$. Note that $|Y|\!\leq\! \frac{p}{\alpha_2}$. Let $X'$ and $Y'$ denote the indices of the nodes in $X$ and $Y$, respectively. Since ${\cal F}$ is a $(|V|,\!k\!+\!\frac{p}{\alpha_2})$-universal set, there is $f\!\in\!{\cal F}$ such that: (1) for all $i\!\in\!X'$, $f(i)\!=\!0$, and (2) for all $i\!\in\!Y'$, $f(i)\!=\!1$. Thus, in $G_f$, the copies of the nodes in $X$ are red, and the copies of the nodes in $Y$ are blue. We get that $(G_f,\!k,\!p)$ is a yes-instance of {\sc NC-$\Pi$}.\qed
\end{proof}

\mysubsection{A Procedure for {\sc NC-$\Pi$}}\label{section:ncpos}

We now present \alg{SolveNCP}, a dynamic programming-based procedure for solving {\sc NC-$\Pi$} (see the pseudocode below). Procedure \alg{SolveNCP} first computes the node-sets of the maximal connected red components in $G$. Then, procedure \alg{SolveNCP} generates a matrix M, where each entry $[i,j]$ holds the minimum value val$(X)$ of a subset $X\!\subseteq\! V$ in $Sol_{i,j}$, the family containing  every set of exactly $j$ nodes in $V$ obtained by choosing a union of sets in $\{C_1,C_2\ldots,C_i\}$, i.e., $Sol_{i,j}\!=\!\{(\bigcup{\cal C}'): {\cal C}'\subseteq \{C_1,C_2,\ldots,C_i\},|\bigcup{\cal C}'|=j\}$. Procedure \alg{SolveNCP} computes M by using dynamic programming, assuming an access to a non-existing entry returns $\infty$, and accepts iff M$[t,k]\!\leq\!p$.

\renewcommand{\thealgorithm}{}
\floatname{algorithm}{Procedure}
\begin{algorithm}[!ht]
\caption{\alg{SolveNCP}($G\!=\!(V,E),k,p$)}
\begin{algorithmic}[1]
\STATE\label{step:ncp1} use DFS to compute the family ${\cal C}\!=\!\{C_1,C_2,\ldots,C_t\}$, for some $0\!\leq\! t\!\leq\! |V|$, of the node-sets of the maximal connected red components in $G$.

\STATE let M be a matrix containing an entry $[i,j]$ for all $0\!\leq\! i\!\leq\! t$ and $0\!\leq\! j\!\leq\! k$.
\STATE initialize M$[i,0]\Leftarrow 0$ for all $0\!\leq\!i\!\leq\! t$, and M$[0,j]\Leftarrow\infty$ for all $1\!\leq\!j\!\leq\!k$.

\FOR{$i\!=\!1,2,\ldots,t,$ and $j\!=\!1,2,\ldots,k$}
	\STATE M$[i,j]\Leftarrow\min\{\mathrm{M}[i\!-\!1,j],\mathrm{M}[i\!-\!1,j\!-\!|C_i|]+\mathrm{val}(C_i)\}$.
\ENDFOR

\STATE {\bf accept} iff M$[t,k]\!\leq\!p$.
\end{algorithmic}
\end{algorithm}

{\noindent The following lemma states the correctness and running time of \alg{SolveNCP}.}

\begin{lemma}\label{lemma:ncp}
Procedure \alg{SolveNCP} solves {\sc NC-$\Pi$} in time $O(|V|k\!+\!|E|)$.
\end{lemma}

\begin{proof}
For all $0\!\leq\! i\!\leq\! t$ and $0\!\leq\! j\!\leq\! k$, denote val$(i,\!j)\!=\!\min_{X\in Sol_{i,j}}\{\mathrm{val}(X)\}$. Using a simple induction on the computation of M, we get that M$[i,\!j]\!=\!\mathrm{val}(i,\!j)$. Since $(G,\!k,\!p)$ is a yes-instance of {\sc NC-$\Pi$} iff val$(t,\!k)\!\leq\! p$, we have that \alg{SolveNCP} is correct. Step \ref{step:ncp1}, and the computation of val$(C)$ for all $C\!\in\!{\cal C}$, are performed in time $O(|V|\!+\!|E|)$. Since M is computed in time $O(|V|k)$, we have that \alg{SolveNCP} runs in time $O(|V|k\!+\!|E|)$.\qed
\end{proof}

\mysubsection{An Algorithm for $\Pi$}\label{section:posalg}

We now conclude \alg{PAlg}, our algorithm for {\sc $\Pi$} (see the pseudocode below). Algorithm \alg{PAlg} calls \alg{ColorNodes} to compute several instances of {\sc NC-$\Pi$}, and accepts iff \alg{SolveNCP} accepts at least one of them.

\renewcommand{\thealgorithm}{3}
\floatname{algorithm}{Algorithm}
\begin{algorithm}[!ht]
\caption{\alg{PAlg}($G\!=\!(V,E),k,p$)}
\begin{algorithmic}[1]
\STATE ${\cal I}\Leftarrow$ \alg{ColorNodes}($G,k,p,k+\frac{p}{\alpha_2}$).

\FORALL{$(G',k',p')\in{\cal I}$}
	\STATE {\bf if} \alg{SolveNCP}$(G',k',p')$ accepts {\bf then} {\bf accept}. {\bf end if}
\ENDFOR

\STATE {\bf reject}.
\end{algorithmic}
\end{algorithm}

{\noindent By Lemmas \ref{lemma:colornodestime}, \ref{lemma:colornodesPcor} and \ref{lemma:ncp}, we have the following result.}
\begin{theorem}
Algorithm \alg{PAlg} solves {\sc $\Pi$} in time $O(2^{k+\frac{p}{\alpha_2}\!+\!o(k+p)}(|V|\!+\!|E|)\!\log\!|V|)$.
\end{theorem}

\mysection{Solving a Subclass of Positive Min-LGPPs Faster}\label{section:res3}

Let $\Pi$ be a min-FGPP satisfying $\alpha_2\!\geq\!\frac{\alpha_1}{2}\!>\!0$. Denote $x\!=\!\max\{\frac{p}{\alpha_2},\min\{\frac{p}{\alpha_1},\frac{p}{\alpha_2}\!+\!(1\!-\!\frac{\alpha_1}{\alpha_2})k\}\}$. In this section we develop an $O^*(2^{x+o(x)})$ time algorithm for {\sc $\Pi$}, that is faster than the algorithm in Section \ref{section:res2b}. Applying a divide-and-conquer step to the edges in the input graph $G$, Section \ref{section:newpsimple} shows that
 we can focus on an easier version of {\sc $\Pi$}. This version is solved in Section \ref{section:ecnewp}
 by using dynamic programming. We give the algorithm in  Section \ref{section:newpalg}.

\mysubsection{Simplifying the Non-Degrading Positive Min-FGPP $\Pi$}\label{section:newpsimple}

We first define an easier variant of {\sc $\Pi$}. Suppose we are given a graph $G\!=\!(V,\!E)$ in which each edge is either red or blue, and parameters $k\in\mathbb{N}$ and $p\in\mathbb{R}$. For any subset $X\!\subseteq\!V$, let C$(X)$ denote the family containing the node-sets of the maximal connected components in the graph $G_r\!=\!(X,E_r)$, where $E_r$ is the set of red edges in $E$ having both endpoints in $X$. Also, let val$^*(X)\!=\!\sum_{C\in\mathrm{C}(X)}\mathrm{val}(C)$. The variant {\sc EC-$\Pi$} asks if there is a subset $X\!\subseteq\! V$ of exactly $k$
nodes, such that all the edges in $E(X,V\!\setminus\! X)$ are blue, and val$^*(X)\!\leq\!p$.

We now present a procedure, called \alg{ColorEdges} (see the pseudocode below), whose input is an instance ($G,\!k,\!p$) of {\sc $\Pi$}. Procedure \alg{ColorEdges} uses a universal set to perform several iterations coloring the edges in $G$, and then returns the resulting set of instances of {\sc EC-$\Pi$}.

\renewcommand{\thealgorithm}{}
\floatname{algorithm}{Procedure}
\begin{algorithm}[!ht]
\caption{\alg{ColorEdges}($G\!=\!(V,E),k,p$)}
\begin{algorithmic}[1]
\STATE let $E\!=\!\{e_1,e_2,\ldots,e_{|E|}\}$.
\STATE\label{step:coloredges2} ${\cal F}\Leftarrow$ \alg{UniSetAlg}$(|E|,x)$.

\FORALL{$f\!\in\!{\cal F}$}
	\STATE let $E^f\!=\!\{e^f_1,e^f_2,\ldots,e^f_{|E|}\}$, where $e^f_i$ is a copy of $e_i$.
	\FOR{$i=1,2,\ldots,|E|$}
		\STATE {\bf if} $f(i)\!=\!0$ {\bf then} color $e^f_i$ red. {\bf else} color $e^f_i$ blue. {\bf end if}
	\ENDFOR
\ENDFOR

\STATE {\bf return} ${\cal I}=\{(G_f\!=\!(V,E_f),k,p): f\!\in\!{\cal F}\}$.
\end{algorithmic}
\end{algorithm}

The following lemma states the correctness of \alg{ColorEdges}.

\begin{lemma}\label{lemma:coloredgescor}
An instance $(G,\!k,\!p)$ of {\sc $\Pi$} is a yes-instance iff \alg{ColorEdges}($G,\!k,\!p$) returns a set ${\cal I}$ containing at least one yes-instance of {\sc EC-$\Pi$}.
\end{lemma}

\begin{proof}
Since $\alpha_2\!\geq\!\frac{\alpha_1}{2}$, val$^*(X)\!\geq\!\mathrm{val}(X)$ for any set $X\!\subseteq\!V$ and coloring of edges in $E$. Thus, if $(G,\!k,\!p)$ is a no-instance of {\sc $\Pi$}, then clearly, for any coloring of edges in $E$, we get a no-instance of {\sc EC-$\Pi$}. Next suppose that $(G,\!k,\!p)$ is a yes-instance, and let $X$ be a set of $k$ nodes in $V$ such that val$(X)\!\leq\!p$. Let $\widetilde{E}_r\!=\!E(X)$, and $E_b\!=\!E(X,\!V\!\setminus\! X)$. Also, choose a minimum-size subset $E_r\!\subseteq\!\widetilde{E}_r$ such that the graphs $G_r'\!=\!(X,\widetilde{E}_r)$ and $G_r\!=\!(X,E_r)$ contain the same set of maximal connected components. Let $E_r'$ and $E_b'$ denote the indices of the edges in $E_r$ and $E_b$, respectively. Note that $|E_r'|\!+\!|E_b'|\!\leq\!x$. Since ${\cal F}$ is an $(|E|,x)$-universal set, there is $f\!\in\!{\cal F}$ such that: (1) for all $i\!\in\!E_r'$, $f(i)\!=\!0$, and (2) for all $i\!\in\!E_b'$, $f(i)\!=\!1$. Thus, in $G_f$, the copies of the edges in $E_r$ are red, and the copies of the edges in $E_b$ are blue. Then, val$^*(X)\!=\!\mathrm{val}(X)$. We get that $(G_f,\!k,\!p)$ is a yes-instance of {\sc EC-$\Pi$}.\qed
\end{proof}
Furthermore, Lemma \ref{lemma:splitter} immediately implies the following result.

\begin{lemma}\label{lemma:coloredgestime}
Procedure \alg{ColorEdges} runs in time $O(2^{x+o(x)}|E|\!\log\!|E|)$, and returns a set ${\cal I}$ of size $O(2^{x+o(x)}\!\log\!|E|)$.
\end{lemma}

\mysubsection{A Procedure for {\sc EC-$\Pi$}}\label{section:ecnewp}

By modifying the procedure given in Section \ref{section:ncpos}, we get a procedure, called \alg{SolveECP}, satisfying the following result (see Appendix \ref{app:ecnewp}).

\begin{lemma}\label{lemma:ecnewp}
Procedure \alg{SolveECP} solves {\sc EC-$\Pi$} in time $O(|V|k\!+\!|E|)$.
\end{lemma}

\mysubsection{A Faster Algorithm for $\Pi$}\label{section:newpalg}

Our faster algorithm for $\Pi$, 
 \alg{FastPAlg}, calls \alg{ColorEdges} to compute several instances of {\sc EC-$\Pi$}, and accepts iff \alg{SolveECP}
 accepts at least one of them (see the pseudocode below).

\renewcommand{\thealgorithm}{4}
\floatname{algorithm}{Algorithm}
\begin{algorithm}[!ht]
\caption{\alg{FastPAlg}($G\!=\!(V,E),k,p$)}
\begin{algorithmic}[1]
\STATE ${\cal I}\Leftarrow$ \alg{ColorEdges}($G,k,p$).

\FORALL{$(G',k',p')\in{\cal I}$}
	\STATE {\bf if} \alg{SolveECP}$(G',k',p')$ accepts {\bf then} {\bf accept}. {\bf end if}
\ENDFOR

\STATE {\bf reject}.
\end{algorithmic}
\end{algorithm}

By Lemmas \ref{lemma:coloredgescor}, \ref{lemma:coloredgestime} and \ref{lemma:ecnewp}, we have the following result.

\begin{theorem}
Algorithm \alg{FastPAlg} solves {\sc $\Pi$} in time $O(2^{x+o(x)}(|V|k\!+\!|E|)\!\log\!|E|)$.
\end{theorem}

{\noindent Since {\sc Min $k$-Vertex Cover} satisfies $\alpha_1\!=\!\alpha_2\!=\!1$, we have the following result.}

\begin{corollary}
Algorithm \alg{FastPAlg} solves {\sc Min $k$-Vertex Cover} in time\\$O(2^{p+o(p)}(|V|k\!+\!|E|)\!\log
\!|E|)$.
\end{corollary}

\bibliographystyle{splncs03}
\bibliography{references}

\begin{thebibliography}{10}
\providecommand{\url}[1]{\texttt{#1}}
\providecommand{\urlprefix}{URL }

\bibitem{B06}
Berkhin, P.: A survey of clustering data mining techniques. Grouping
  Multidimensional Data Recent Advances in Clustering, Eds. J. Kogan and C.
  Nicholas and M. Teboulle pp. 25--71 (2006)

\bibitem{localpartition}
Bonnet, E., Escoffier, B., Paschos, V.T., Tourniaire, E.: Multi-parameter
  complexity analysis for constrained size graph problems: using greediness for
  parameterization. In: IPEC. pp. 66--77 (2013)

\bibitem{kdensest}
Bourgeois, N., Giannakos, A., Lucarelli, G., Milis, I., Paschos, V.T.: Exact
  and approximation algorithms for densest $k$-subgraph. In: IPEC. pp. 66--77
  (2013)

\bibitem{cai08}
Cai, L.: Parameterized complexity of cardinality constrained optimization
  problemss. Comput. J.  51(1),  102�--121 (2008)

\bibitem{randsep}
Cai, L., Chan, S.M., Chan, S.O.: Random separation: A new method for solving
  fixed-cardinality optimization problems. In: IPEC. pp. 239--250 (2006)

\bibitem{bisection}
Cygan, M., Lokshtanov, D., Pilipczuk, M., Pilipczuk, M., Saurabh, S.: Minimum
  bisection is fixed parameter tractable. In: STOC (2014, to appear)

\bibitem{DRLJ10}
Donavalli, A., Rege, M., Liu, X., Jafari-Khouzani, K.: Low-rank matrix
  factorization and co-clustering algorithms for analyzing large data sets. In:
  ICDEM. pp. 272--279 (2010)

\bibitem{downey03}
Downey, R.G., Estivill-Castro, V., Fellows, M.R., Prieto, E., Rosamond, F.A.:
  Cutting up is hard to do: the parameterized complexity of $k$-cut and related
  problems. Electr. Notes Theor. Comput. Sci.  78,  209--222 (2003)

\bibitem{downey95}
Downey, R.G., Fellows, M.R.: Fixed-parameter tractability and completeness
  $\mathrm{II}$: on completeness for $\mathrm{W}$[1]. Theor. Comput. Sci.
  141(1\&2),  109--131 (1995)

\bibitem{representative}
Fomin, F.V., Lokshtanov, D., Saurabh, S.: Efficient computation of
  representative sets with applications in parameterized and exact agorithms.
  In: SODA. pp. 142--151 (2014)

\bibitem{vcvariants}
Guo, J., Niedermeier, R., Wernicke, S.: Parameterized complexity of vertex
  cover variants. Theory Comput. Syst.  41(3),  501--520 (2007)

\bibitem{KLM11}
Kahng, A.B., Lienig, J., Markov, I.L., Hu, J.: VLSI Physical Design - From
  Graph Partitioning to Timing Closure. Springer (2011)

\bibitem{kloks94}
Kloks, T.: Treewidth, computations and approximations. LNCS 842, Springer
  (1994)

\bibitem{partialvc}
Kneis, J., Langer, A., Rossmanith, P.: Improved upper bounds for partial vertex
  cover. In: WG. pp. 240--251 (2008)

\bibitem{enuconcom}
Komusiewicz, C., Sorge, M.: Finding dense subgraphs of sparse graphs. In: IPEC.
  pp. 242--251 (2012)

\bibitem{splitter}
Naor, M., Schulman, L.J., Srinivasan, A.: Splitters and near-optimal
  derandomization. In: FOCS. pp. 182--191 (1995)

\bibitem{corrrepresentative}
Shachnai, H., Zehavi, M.: Faster computation of representative families for
  uniform matroids with applications. CoRR abs/1402.3547  (2014)

\bibitem{corrmatchpack}
Zehavi, M.: Deterministic parameterized algorithms for matching and packing
  problems. CoRR abs/1311.0484  (2013)

\end{thebibliography}

\newpage

\appendix

\mysection{An Illustration of the Reduction $f$}\label{app:fig1}

\begin{figure}
\medskip
\centering
\includegraphics[scale=0.5]{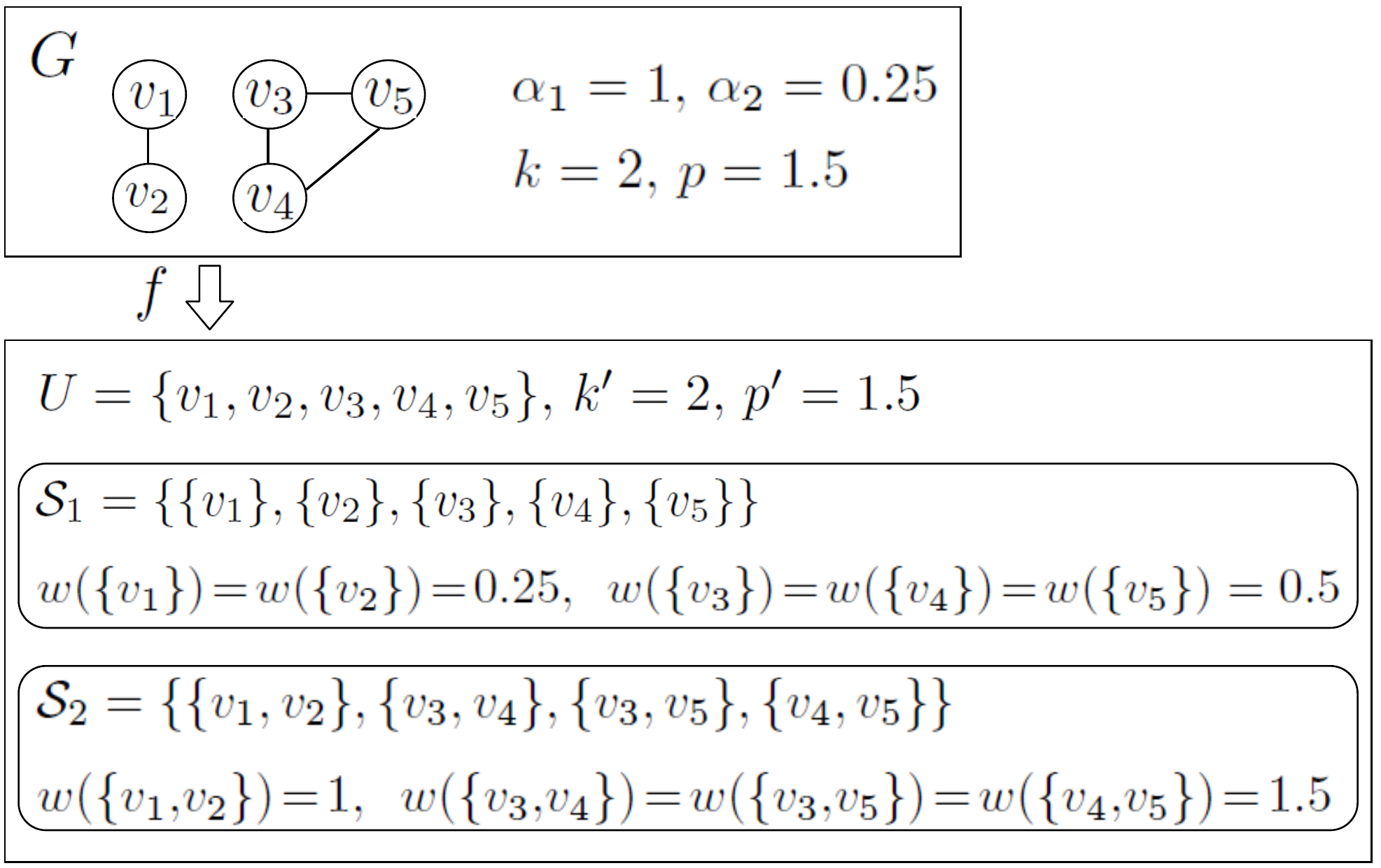}
\caption{An illustration of the reduction $f$, given in Section \ref{section:reduce}.}
\label{fig:1}
\end{figure}

\mysection{A Procedure for {\sc EC-$\Pi$} (Cont.)}\label{app:ecnewp}

We now present the details of procedure \alg{SolveECP} (see the pseudocode below). Procedure \alg{SolveECP} first computes the node-sets of the maximal connected components in the graph obtained by removing all the blue edges from $G$. Then, procedure \alg{SolveECP} generates a matrix M, where each entry $[i,j]$ holds the minimum value val$^*(X)$ of a subset $X\!\subseteq\! V$ in $Sol_{i,j}$, the family containing  every set of exactly $j$ nodes in $V$ obtained by choosing a union of sets in $\{C_1,C_2\ldots,C_i\}$, i.e., $Sol_{i,j}\!=\!\{(\bigcup{\cal C}'): {\cal C}'\subseteq \{C_1,C_2,\ldots,C_i\},|\bigcup{\cal C}'|=j\}$. Procedure \alg{SolveNCP} computes M by using dynamic programming, assuming an access to a non-existing entry returns $\infty$, and accepts iff M$[t,k]\!\leq\!p$.

\renewcommand{\thealgorithm}{}
\floatname{algorithm}{Procedure}
\begin{algorithm}[!ht]
\caption{\alg{SolveECP}($G\!=\!(V,E),k,p$)}
\begin{algorithmic}[1]
\STATE\label{step:ecnewp1} use DFS to compute the family ${\cal C}\!=\!\{C_1,C_2,\ldots,C_t\}$, for some $0\!\leq\! t\!\leq\! |V|$, of the node-sets of the maximal connected components in the graph obtained by removing all the blue edges from $G$.

\STATE let M be a matrix containing an entry $[i,j]$ for all $0\!\leq\! i\!\leq\! t$ and $0\!\leq\! j\!\leq\! k$.
\STATE initialize M$[i,0]\Leftarrow 0$ for all $0\!\leq\!i\!\leq\! t$, and M$[0,j]\Leftarrow\infty$ for all $1\!\leq\!j\!\leq\!k$.

\FOR{$i\!=\!1,2,\ldots,t,$ and $j\!=\!1,2,\ldots,k$}
	\STATE M$[i,j]\Leftarrow\min\{\mathrm{M}[i\!-\!1,j],\mathrm{M}[i\!-\!1,j\!-\!|C_i|]+\mathrm{val}^*(C_i)\}$.
\ENDFOR

\STATE {\bf accept} iff M$[t,k]\!\leq\!p$.
\end{algorithmic}
\end{algorithm}

{\noindent We next prove the correctness of Lemma \ref{lemma:ecnewp}.}

\begin{proof}
For all $0\!\leq\! i\!\leq\! t$ and $0\!\leq\! j\!\leq\! k$, denote val$(i,\!j)\!=\!\min_{X\in Sol_{i,j}}\{\mathrm{val}^*(X)\}$. Using a simple induction on the computation of M, we get that M$[i,\!j]\!=\!\mathrm{val}(i,\!j)$. Since $(G,\!k,\!p)$ is a yes-instance of {\sc EC-$\Pi$} iff val$(t,\!k)\!\leq\! p$, we have that \alg{SolveECP} is correct. Step \ref{step:ecnewp1}, and the computation of val$^*(C)$ for all $C\!\in\!{\cal C}$, are performed in time $O(|V|\!+\!|E|)$. Since M is computed in time $O(|V|k)$, we have that \alg{SolveECP} runs in time $O(|V|k\!+\!|E|)$.\qed
\end{proof}

\mysection{Some Proofs}\label{app:someproofs}
\noindent
{\bf Proof of lemma \ref{lemma:decreasecor}:}
First, assume that $(U,\!{\cal S},\!w,\!k',\!p')$ is a yes-instance. Let ${\cal S}'$ be 
a subfamily of disjoint sets from ${\cal S}$, such that $|\bigcup{\cal S}'|\!=\!k'$, $\sum_{S\in{\cal S}'}w(S)\!\geq\!p'$, 
and there is no subfamily ${\cal S}''$ satisfying these conditions, and
 $|{\cal S}'\!\cap\!\widehat{\cal S}|\!<\!|{\cal S}''\!\cap\!\widehat{\cal S}|$. Suppose, by way of
 contradiction, that there is a set $S\!\in\!({\cal S}_i\cap{\cal S}')\!\setminus\!\widehat{\cal S}$,
 for some $1\!\leq\!i\!\leq\!k'$. By Theorem \ref{theorem:rep}, there is a set $\widehat{S}\!\in\!\widehat{\cal S}_i$
 such that $w(\widehat{S})\!\geq\!w(S)$, and $\widehat{S}\!\cap\!S'\!=\!\emptyset$,
 for all $S'\!\in\!{\cal S}'\!\setminus\!\{S\}$. Thus, ${\cal S}''\!=\!({\cal S}'\!\setminus\!\{S\})\!\cup\!\{\widehat{S}\}$
 is a solution to $(U,\!{\cal S},\!w,\!k',\!p')$. Since $|{\cal S}'\!\cap\!\widehat{\cal S}|\!<\!|{\cal S}''\!\cap\!\widehat{\cal S}|$,
 this is a contradiction.

Now, assume that $(U,\!\widehat{\cal S},\!w,\!k',\!p')$ is a yes-instance. Since $\widehat{\cal S}\!\subseteq\!{\cal S}$, 
we immediately get that $(U,\!{\cal S},\!w,\!k',\!p')$ is also a yes-instance.\qed

\end{document}